\documentclass[12pt]{article}

\usepackage{amsmath,amssymb,amsthm}
\newtheorem{theorem}{Theorem}
\newtheorem{corollary}[theorem]{Corollary}
\newtheorem{definition}[theorem]{Definition}
\newtheorem{lemma}[theorem]{Lemma}

\begin{document}
\title{A State-Dependent Noncontextuality Inequality \\ 
in Algebraic Quantum Theory}
\author{Yuichiro Kitajima}
\maketitle

\begin{abstract}
The noncontextuality condition states that a value of any observable is independent of which other compatible observable is measured jointly with it. Klyachko, Can, Binicio{\u{g}}lu, and Shumovsky have introduced an inequality which holds if there is a noncontextual hidden variable theory. It is called KCBS inequality, which is state-dependent. Its violation shows a contradiction between predictions of quantum theory and noncontextual hidden variable theories. In the present paper, it is shown that there is a state which does not violate KCBS inequality in the case of quantum mechanics of finite degrees of freedom, and that any normal state violates it in the case of algebraic quantum field theory. It is a difference between quantum mechanics of finite degrees of freedom and algebraic quantum field theory from a point of view of KCBS inequality.
\end{abstract}

\section{Introduction}

Quantum theory is a probabilistic theory. It is impossible to predict a definite value of every observable, and only the probability distributions are given. The basic idea of hidden variable theories is to argue that probability arises because quantum states do not represent the ultimate information about the system, and that there are hidden variables by which values of all observables would be uniquely determined.

The necessity of a completion of quantum theory with additional hidden variables was advocated by Einstein, Podolsky, and Rosen \cite{einstein1935can}. They used the locality condition in their argument. It states that, if two measurements events are spacelike separated, a measurement performed on the first system cannot affect a measurement on the second system. Bell subsequently developed their argument, and showed that a hidden variable theory satisfying the locality condition must obey an inequality, which is nowadays known as Bell inequality \cite{bell1964einstein,clauser1969proposed}. This inequality is based on only the locality condition and independent of quantum mechanics. Thus, it is testable in experiments. Its conclusive experimental violation shows a contradiction between testable predictions of quantum theory and local hidden variable theories \cite{aspect1982experimental}. 

Another result of hidden variables in quantum theory is Bell-Kochen-Specker (BKS) theorem \cite{bell1966problem,kochen1967problem,mermin1993hidden}. In this theorem, the noncontextuality condition is used instead of the locality condition. The noncontextuality condition states that a value of any observable is independent of which other compatible observable is measured jointly with it. Two measurements are said to be compatible if they can be performed jointly on the same system without disturbing each other, and a measurement context is then defined as a set of compatible measurement. Bell inequality is a particular type of noncontextuality inequality in which measurements are not only compatible but also spacelike separated. Thus the noncontextuality condition is a generalization of the locality condition. 

BKS theorem shows that there is no noncontextual hidden variable theory by a logical contradiction. This argument is state-independent while Bell inequality is state-dependent. Recently, Klyachko, Can, Binicio{\u{g}}lu, and Shumovsky \cite{klyachko2008simple} have introduced an inequality which holds if there is a noncontextual hidden variable theory. It is called KCBS inequality, which is state-dependent as well as Bell inequality. A difference between them is that KCBS inequality requires not the locality condition but the noncontextuality condition. Its violation shows a contradiction between predictions of quantum theory and noncontextual hidden variable theories. KCBS inequality elevates Bell-Kochen-Specker theorem to experimentally testable propositions because it is based only the noncontextuality condition and independent of quantum mechanics. It can be tested whether KCBS inequality holds or not, and two recent experiments have shown that it does not hold in quantum mechanics \cite{lapkiewicz2011experimental,ahrens2013two}.

In the present paper, we examine KCBS inequality in quantum mechanics of finite degrees of freedom and algebraic quantum field theory. In Section \ref{KCBS-section}, we sketch the reason why KCBS inequality is related to noncontextual hidden variable theories. Its violation shows a contradiction between predictions of quantum theory and noncontextual hidden variable theories. It is known that all pure normal states violate KCBS inequality in quantum mechanics of finite degrees of freedom. In Section \ref{finite-section} it is shown that a tracial state does not violate KCBS inequality in the case of quantum mechanics of finite degrees of freedom. In Section \ref{type-III-section} we examine KCBS inequality in algebraic quantum field theory. A tracial states and a pure normal state do not exist in algebraic quantum field theory because any local algebra is of type III \cite[Section 4]{clifton2001entanglement}. Thus, the argument of KCBS inequality in quantum mechanics of finite degrees of freedom cannot apply to algebraic quantum field theory. It is shown that any normal state violates it in the case of algebraic quantum field theory in Section \ref{type-III-section}. It is a difference between quantum mechanics of finite degrees of freedom and algebraic quantum field theory from a point of view of KCBS inequality.

\section{KCBS inequality}
\label{KCBS-section}
If it is possible to assign values to some observables, there is a value assignment on these observables \cite[Definition 2.2 and Lemma 2.2]{doring2005kochen} \cite[Section 5.1]{redhead1987incompleteness}. A mapping $v$ from a set $\mathbb{A}$ of self-adjoint operators in a von Neumann algebra $\mathfrak{N}$ to the set of all real numbers is called a value assignment on $\mathbb{A}$ if for any mutually commuting self-adjoint operators $A, B \in \mathbb{A}$
\begin{equation}
\label{valuation}
\begin{split}
& v(A+B)=v(A)+v(B), \\
& v(AB)=v(A)v(B), \\
&v(I)=1, \\
&v(0)=0.
\end{split}
\end{equation}
KCBS inequality is derived with a value assignment. There are two forms \cite{cabello2010non}. 

First, consider five self-adjoint operators $A_i$ such that $A_i^2=I$ and $[A_i, A_{i+1}]=0$ for any $i \in \{ 0, \dots ,4 \}$ and $i+1$ is understood mod $5$. Then $A_i$ and $A_{i+1}$ can be jointly measured without mutual disturbance, that is, they are compatible. 
The noncontextuality condition states that a value of any observable is independent of which other compatible observable is measured jointly with it. Thus, a noncontextual hidden variable theory is that in which the value of $A_{i}$ is independent of whether $A_{i}$ is measured together with $A_{i-1}$ or together with $A_{i+1}$ although a context represented by $A_{i-1}$ and $A_{i}$ is not necessarily same as that represented by $A_{i}$ and $A_{i+1}$. 
If such a theory exists, there is a value assignment $v_{1}$ on $\{A_0, \dots, A_4 \}$. By Equation (\ref{valuation}), $v_{1}(A_i)=-1, 1$ and $v_{1}(A_iA_{i+1})=v_{1}(A_i)v_{1}(A_{i+1})$ for any $i \in \{0,\dots, 4 \}$. Thus
\begin{equation}
\label{kcbs-inequality-value}
\sum_{i=0}^{4}v_{1}(A_{i}A_{i+1}) \geq -3.
 \end{equation}
Suppose that for any normal state $\psi$ of $\mathfrak{N}$ there is a hidden variable $\lambda \in \Lambda$ with which the the value assignment  $v_{1}$ is labeled, and a probability measure $\mu$ on the space $\Lambda$ such that
\begin{equation}
\psi(A_{i}A_{i+1})=\sum_{\Lambda} v_{1}(A_{i}A_{i+1} | \lambda) \mu(\lambda)
\end{equation}
for any $i \in \{ 0, \dots, 4 \}$ \cite[p. 697]{bub2009contextuality}. It requires that the statistics of $v_1$ should coincide with the result of applying the statistical algorithm of quantum mechanics. Then
\begin{equation}
\label{kcbs-inequality-value-2}
\begin{split}
\psi \left( \sum_{i=0}^{4} A_{i}A_{i+1} \right) 
&=\sum_{i=0}^{4} \psi \left( A_{i}A_{i+1} \right)  \\
&=\sum_{i=0}^{4} \sum_{\Lambda} v_{1}(A_{i}A_{i+1} | \lambda) \mu(\lambda) \\
&=\sum_{\Lambda}  \sum_{i=0}^{4} v_{1}(A_{i}A_{i+1} | \lambda)  \mu(\lambda) \\
&\geq -3.
\end{split}
\end{equation}

Inequality (\ref{kcbs-inequality-value-2}) is the first form of KCBS inequality \cite{klyachko2008simple,cabello2010non}. Because it does not contain a hidden variable $\lambda$, it is testable in experiments. Two recent experiments have shown that it does not hold in quantum mechanics \cite{lapkiewicz2011experimental,ahrens2013two}.

Next, consider another version of KCBS inequality.  Let $P_0, \dots, P_4$ be five projections in a von Neumann algebra $\mathfrak{N}$ such that $P_iP_{i+1}=0$ for any $i \in \{0, \dots, 4 \}$, where $i+1$ is understood mod $5$. These projections represent 5 yes-no questions such that $P_{i}$ and $P_{i+1}$ are compatible and exclusive. 
In a noncontextual hidden variable theory, the value of $P_i$ is independent of whether $P_i$ is measured together with $P_{i-1}$ or together with $P_{i+1}$. If such a theory exists, there is a value assignment $v_{2}$ on $\{ P_0, \dots, P_4 \}$. By Equation (\ref{valuation}), $v_{2}(P_i)=0,1$ and $v_{2}(P_{i})v_{2}(P_{i+1})=0$ for any $i \in \{ 0, \dots, 4 \}$. Thus
\begin{equation}
\label{kcbs-inequality2-value}
\sum_{i=0}^{4} v_{2}(P_{i}) \leq 2.
\end{equation}
Suppose that for any normal state $\psi$ of $\mathfrak{N}$ there is a hidden variable $\lambda \in \Lambda$ with which the the valuation function $v_{2}$ is labeled, and a probability measure $\mu$ on the space $\Lambda$ such that
\begin{equation}
\psi(P_{i})=\sum_{\Lambda} v_{2} (P_{i} | \lambda) \mu(\lambda)
\end{equation}
for any $i \in \{0, \dots , 4 \}$ \cite[p. 697]{bub2009contextuality}.
It requires that the statistics of $v_2$ should coincide with the result of applying the statistical algorithm of quantum mechanics. Then
\begin{equation}
\label{kcbs-inequality2-value-2}
\begin{split}
\psi \left( \sum_{i=0}^{4} P_{i} \right) &=\sum_{i=0}^{4} \psi (P_{i}) \\
&=\sum_{i=0}^{4}  \sum_{\Lambda} v_{2} (P_{i} | \lambda) \mu(\lambda) \\
&=\sum_{\Lambda}  \sum_{i=0}^{4}  v_{2} (P_{i} | \lambda) \mu(\lambda) \\
&\leq 2.
\end{split}
\end{equation}
Inequality (\ref{kcbs-inequality2-value-2}) is the second form of KCBS inequality \cite{bub2009contextuality,badziag2011pentagrams,cabello2010non,cabello2013simple}. In this paper, we examine this form. A violation of KCBS inequality is defined as follows.

\begin{definition}
\label{KCBS-definition}
Let $\mathfrak{N}$ be a von Neumann algebra, and let $\psi$ be a normal state of $\mathfrak{N}$. We say that $\psi$ violates KCBS inequality if there are projections $P_0, \dots P_4 \in \mathfrak{N}$ such that $P_iP_{i+1}=0$, where $i+1$ is understood mod $5$, and 
\[ \psi(P_{0}+P_{1}+P_{2}+P_{3}+P_{4}) > 2. \]
\end{definition}

KCBS inequality is state-dependent as well as Bell inequality. A difference between them is that KCBS inequality requires not the locality condition but the noncontextuality condition. Its violation shows a contradiction between predictions of quantum theory and noncontextual hidden variable theories.


BKS theorem shows that there are no noncontextual hidden variable theories in algebraic quantum field theory as well as quantum mechanics of finite degrees of freedom \cite{doring2005kochen,kitajima2006problem}. Its approach is different from KCBS inequality. Suppose that a value assignment $v_{3}$ on a set of all projections in a von Neumann algebra $\mathfrak{N}$ which has neither direct summand of type $I_1$ nor direct summand of type $I_2$. By generalized Gleason theorem \cite{maeda1989probability}, $v_{3}$ can be extended to a dispersion-free state of $\mathfrak{N}$ \cite[Lemma 5.1]{hamhalter1993pure}. But this state does not exist in $\mathfrak{N}$. Thus a value assignment $v_{3}$ does not exist in $\mathfrak{N}$ \cite[Theorem 5.2]{kitajima2006problem}. 

There are two differences between KCBS inequality and BKS theorem. First, the sets on which the value assignments $v_1$, $v_2$, and $v_3$ are defined are different. Second, KCBS inequality cannot be derived without a hidden variable $\lambda \in \Lambda$ with which the the value assignment is labeled and a probability measure $\mu$ on the space $\Lambda$, while they do not play an important role in BKS theorem. Thus, BKS theorem is a different argument from KCBS inequality although their conclusions are same. A violation of KCBS inequality reinforces the conclusion of BKS theorem.


While a C*-algebra does not necessarily contain a non-zero projection, a von Neumann algebra contains it. Thus, we examine a von Neumann algebra because KCBS inequality cannot be defined without five non-zero projections. There will, in general, be many non-normal states of a von Neumann algebra $\mathfrak{N}$. By Fell's theorem \cite{fell1960dual} \cite[p. 428]{clifton2001rindler}, any non-normal state $\psi'$ of $\mathfrak{N}$ can be weak* approximated by a normal state of $\mathfrak{N}$, that is, for any real number $\epsilon > 0$, and for each finite collection $\{ A_i \in \mathfrak{N} | i=1, \dots, n \}$ of operators, there is a normal state $\psi$ of $\mathfrak{N}$ such that
\begin{equation}
| \psi'(A_i) - \psi(A_i) | < \epsilon \ \ \ \ (i=1, \dots, n).
\end{equation}
It means that a finite number of measurements with finite accuracy cannot distinguish $\psi$ and $\psi'$. Therefore, we examine only normal states of a von Neumann algebra.

\section{The case of quantum mechanics of finite degrees of freedom}
\label{finite-section}
In this section, we examine KCBS inequality in a finite von Neumann algebra because quantum mechanics of finite degrees of freedom is described with a finite dimensional Hilbert space, and the set of all operators on a finite dimensional Hilbert space is a finite von Neumann algebra. In a 3-dimensional Hilbert space, there is a state which violates KCBS inequality. Let define the following unit vectors.

\begin{equation}
\label{KCBS-Cabello}
\begin{split}
&\Psi_0= \frac{1}{\sqrt{1+\cos (1/5)\pi}}\begin{pmatrix} 1 \\ 0 \\ \sqrt{\cos (1/5)\pi} \end{pmatrix}, \ \ \ 
\Psi_1=\frac{1}{\sqrt{1+\cos (\pi/5)}}\begin{pmatrix} \cos (4/5)\pi \\ \sin (4/5)\pi \\ \sqrt{\cos (1/5)\pi} \end{pmatrix}, \\
&\Psi_2=\frac{1}{\sqrt{1+\cos (\pi/5)}}\begin{pmatrix} \cos (2/5)\pi \\ -\sin(2/5)\pi \\ \sqrt{\cos (1/5)\pi} \end{pmatrix}, \ \ \
\Psi_3=\frac{1}{\sqrt{1+\cos (\pi/5)}}\begin{pmatrix} \cos (2/5)\pi \\ \sin(2/5)\pi \\ \sqrt{\cos (1/5)\pi} \end{pmatrix}, \\
&\Psi_4=\frac{1}{\sqrt{1+\cos (\pi/5)}}\begin{pmatrix} \cos (4/5)\pi \\ -\sin(4/5)\pi \\ \sqrt{\cos (1/5)\pi} \end{pmatrix}, \ \ \
\Psi=\begin{pmatrix} 0 \\ 0 \\ 1 \end{pmatrix}.
\end{split}
\end{equation}

Let $R_i$ be the projection whose range is the subspace generated by $\Psi_i$ ($i=0, \dots ,4$). Then
$R_{i}R_{i+1}=0$ for any $i \in \{0, \dots, 4 \}$, where $i+1$ is understood mod $5$, and
\begin{equation}
\label{violation-KCBS-inequality}
 \langle \Psi, (R_{0}+R_{1}+R_{2}+R_{3}+R_{4} ) \Psi \rangle =  \frac{5 \cos (1/5)\pi}{1+\cos (1/5)\pi} = \sqrt{5} > 2. 
 \end{equation}
 
 Thus, the vector state induced by $\Psi$ violates KCBS inequality \cite{bub2009contextuality,badziag2011pentagrams,cabello2010non,cabello2013simple}. 
 
On the other hand, a mixed state 
\begin{equation}
 \tau_{3}(\cdot)=\text{Tr} \left( \left( \frac{1}{3}I \right) \cdot  \right) 
\end{equation}
  does not violate KCBS inequality in a 3-dimensional Hilbert space, where $I$ is an identity operator. In other words,
\begin{equation}
\text{Tr} \left( \left( \frac{1}{3}I \right) \left( P_{0} +P_{1}+P_{2}+P_{3} +P_{4} \right) \right) \leq 2
\end{equation}
for any non-zero projections $P_{i}$ such that $P_iP_{i+1}=0$ for any $i \in \{0, \dots, 4 \}$, where $i+1$ is understood mod $5$,
because $\text{Tr}(P_{i})$ is the dimension of $P_{i}$ and the number of two-dimensional projections is less than $2$. This result can be extended to a tracial state. 

\begin{definition} \cite[p.505]{kadison1997fundamentals}
Let $\mathfrak{N}$ be a finite von Neumann algebra, and let $\tau$ be a state of $\mathfrak{N}$. $\tau$ is called a tracial state if $\tau(AB)=\tau(BA)$ for any $A,B \in \mathfrak{N}$.
\end{definition}

\begin{theorem}
\label{trace-state}
Let $\mathfrak{N}$ be a finite von Neumann algebra, and let $\tau$ be a tracial state of $\mathfrak{N}$. Then $\tau$ does not violate KCBS inequality.
\end{theorem}

\begin{proof}
Let $P_0, \dots, P_4$ be projections in $\mathfrak{N}$ such that $P_iP_{i+1}=0$ for any $i \in \{0, \dots, 4 \}$, where $i+1$ is understood mod $5$. 
Then
\begin{equation}
\label{e1_finite}
 1 - \tau(P_0 \wedge P_3) \geq \tau(P_1) + \tau(P_2) 
 \end{equation}
since $I - P_0 \wedge P_3 \geq I - P_0 \geq P_1$ and $I - P_0 \wedge P_3 \geq I - P_3 \geq P_2$, and $P_1P_2=0$. Because the set of all projections of $\mathfrak{N}$ is  an orthomodular poset \cite[Proposition 1.3]{maeda1989probability}, $P_0=(P_0 \wedge P_3) \vee (P_0 \wedge (P_0 \wedge P_3)^{\perp})$.
This equation and $(P_0 \wedge P_3) \perp (P_0 \wedge (P_0 \wedge P_3)^{\perp})$ imply
\begin{equation}
\label{e2_finite}
\tau(P_0)=\tau(P_0 \wedge P_3) + \tau(P_0 \wedge (P_0 \wedge P_3)^{\perp}).
\end{equation}
Note that
\begin{equation}
\label{trace-additivity}
\tau(P) + \tau(Q) = \tau(P \vee Q) + \tau(P \wedge Q)
\end{equation}
for any $P,Q \in \mathfrak{N}$ because there is an isometry $V \in \mathfrak{N}$ such that $P \vee Q - Q =VV^{*}$ and $P - P \wedge Q=V^{*}V$ \cite[Theorem 6.1.7]{kadison1997fundamentals}.
Equation (\ref{trace-additivity}), $(P_0 \wedge (P_0 \vee P_3)^{\perp}) \wedge P_3 = 0$, and $(P_0 \wedge (P_0 \wedge P_3)^{\perp}) \vee P_3 \leq P_0 \vee P_3$ entail
\begin{equation}
\label{e3_finite}
\begin{split}
&\tau(P_0 \wedge (P_0 \wedge P_3)^{\perp}) + \tau(P_3) \\
&=\tau((P_{0} \wedge (P_{0} \vee P_{3})^{\perp}) \vee P_{3})+\tau((P_{0} \wedge (P_{0} \wedge P_{3})^{\perp}) \wedge P_{3}) \\
&=\tau((P_{0} \wedge (P_{0} \vee P_{3})^{\perp}) \vee P_{3}) \\
&\leq \tau(P_0 \vee P_3).
\end{split}
\end{equation}

Thus
\begin{equation}
\label{dimension-inequality-finite}
 \begin{split}
&\tau(P_0)+\tau(P_1)+\tau(P_2)+\tau(P_3)+\tau(P_4) \\
&\leq \tau(P_0) + 1- \tau(P_0 \wedge P_3)+\tau(P_3)+\tau(P_4) \ \ \ (\because \text{Inequality (\ref{e1_finite})}) \\
&=1+\tau(P_0 \wedge (P_0 \wedge P_3)^{\perp})+\tau(P_3)+\tau(P_4) \ \ \ (\because \text{Equation (\ref{e2_finite})}) \\
&\leq 1+\tau(P_0 \vee P_3)+\tau(P_4) \ \ \ (\because \text{Inequality (\ref{e3_finite})}) \\
&\leq 2 \ \ \ (\because P_0 \vee P_3 \perp P_4).
\end{split} 
\end{equation}

\end{proof}

Let $\mathbb{B}(\mathcal{H}_n)$ be the set of all operators on a $n$-dimensional Hilbert space $\mathcal{H}_n$, and let $\tau_{n}$ be a state of $\mathbb{B}(\mathcal{H}_n)$ such that
\begin{equation}
\label{trace}
 \tau_{n}(A)=\text{Tr} \left( \left( \frac{1}{n}I \right) A \right) 
\end{equation}
for any $A \in \mathbb{B}(\mathcal{H}_n)$, where $I$ is an identity operator on $\mathcal{H}_n$. Since $\tau_{n}$ is a tracial state of $\mathbb{B}(\mathcal{H}_n)$ \cite[Example 8.1.2]{kadison1997fundamentals}, $\tau_{n}$ does not violate KCBS inequality.  

In the case of the set of all operators on a 2-dimensional Hilbert space $\mathcal{H}_2$, Theorem \ref{trace-state} and Equation (\ref{trace}) entail
\begin{equation}
\label{kcbs-trace-inequality}
\frac{1}{2}\text{Tr}(P_0+P_1+P_2+P_3+P_4) \leq 2
\end{equation}
for any projections $P_0, \dots, P_4 \in \mathbb{B}(\mathcal{H}_2)$ such that $P_iP_{i+1}=0$ for any $i \in \{0, \dots, 4 \}$, where $i+1$ is understood mod $5$. By Inequality (\ref{kcbs-trace-inequality}),
\begin{equation}
d(P_{0})+d(P_{1})+d(P_{2})+d(P_{3})+d(P_{4}) \leq 4,
\end{equation}
where $d(P_{i})$ is the dimension of $P_{i}\mathcal{H}_{2}$ ($i=0, \dots, 4$). Thus, there is a projection $P_j \in \{ P_0, \dots , P_4 \}$ such that $P_j=0$. Then, for any unit vector $\Psi \in \mathcal{H}_{2}$,
\begin{equation}
\langle \Psi, (P_0+P_1+P_2+P_3+P_4) \Psi \rangle \leq 2 
\end{equation}
since $P_iP_{i+1}=0$.
Therefore any state of $\mathbb{B}(\mathcal{H}_2)$ does not violate KCBS inequality \cite[Observation 1]{guhne2014bounding}.
It is similar to a restriction of Bell-Kochen-Specker Theorem because this theorem cannot apply to $\mathbb{B}(\mathcal{H}_2)$.

Theorem \ref{trace-state} and the following part show that there is a state which does not violate KCBS inequality in the case of quantum mechanics of finite degrees of freedom. 

\section{The case of algebraic quantum field theory}
\label{type-III-section}
Bell inequality has been investigated in algebraic quantum field theory as well as quantum mechanics of finite degrees of freedom \cite{halvorson2000generic,kitajima2013epr,landau1987violation,summers1987bell,summers1987abell,summers1987maximal,summers1988maximal}. For example, it is shown that many normal states violate Bell inequality for any two spacelike separated regions \cite{halvorson2000generic,kitajima2013epr}.



In this section, we examine KCBS inequality in algebraic quantum field theory. Although this inequality is state-dependent as well as Bell inequality, it does not require spacelike separated measurements. The argument about KCBS inequality is made in one region rather than two spacelike separated regions. It is a difference between Bell inequality and KCBS inequality.

Algebraic quantum field theory exists in two versions: the Haag-Araki theory which uses von Neumann algebras on a Hilbert space, and the Haag-Kastler theory which uses abstract C*-algebras. The Haag-Kastler theory looks different from the traditional formalism of quantum mechanics because it lacks a Hilbert space. But GNS representation theorem shows that a Hilbert space is hidden inside a C*-algebra. In other words, by GNS representation of a C*-algebra $\mathfrak{A}$, we can get a Hilbert space and self-adjoint operators which are corresponding to self-adjoint elements in $\mathfrak{A}$. In the Haag-Araki theory, to get the traditional formalism, a representation induced by a state is chosen. One of important representations is a vacuum representation. For example, Doplicher-Haag-Roberts developed the superselection theory of this representation \cite{baumgartel1995operatoralgebraic, halvorson2006algebraic}.

Here we adopt the Haag-Araki theory and a vacuum representation. In this theory, each bounded open region $\mathcal{O}$ in the Minkowski space is associated with a von Neumann algebra $\mathfrak{N}(\mathcal{O})$ on a Hilbert space $\mathcal{H}$. Such a von Neumann algebra is called a local algebra. Under usual axioms, the vacuum vector is a cyclic and separating vector for all local algebras \cite[Corollary 1.3.3]{baumgartel1995operatoralgebraic}, and a typical local algebra is a type III factor \cite[Section V.2]{haag1996local}. Moreover, the funnel property is sometimes assumed. It asserts that there exists a type I factor $\mathfrak{N}$ such that $\mathfrak{N}(\mathcal{O}) \subset \mathfrak{N} \subset \mathfrak{N}(\tilde{\mathcal{O}})$ for any pair $(\mathcal{O}, \tilde{\mathcal{O}})$ of double cones in the Minkowski space such that $\bar{\mathcal{O}} \subset \tilde{\mathcal{O}}$, where $\bar{\mathcal{O}}$ is the closure of $\mathcal{O}$ \cite[Section 2.4]{halvorson2006algebraic}. If the net of local algebras on a Hilbert space $\mathcal{H}$ satisfies the funnel property and the vacuum vector is cyclic and separating for any local algebra, $\mathcal{H}$ is separable \cite[Proposition 63]{halvorson2006algebraic}. Thus, we examine a type III factor on a separable Hilbert space in this section.

In Theorem \ref{trace-state}, it was shown that a tracial state of a finite von Neumann algebra does not violate KCBS inequality.
Its counterpart does not exist in algebraic quantum field theory because there is no tracial state of a type III factor. Moreover, there is no pure normal state in a type III factor \cite[Section 4]{clifton2001entanglement} while any pure normal state violates KCBS inequality in quantum mechanics of finite degrees of freedom. Thus, the argument of KCBS inequality in quantum mechanics of finite degrees of freedom cannot apply to algebraic quantum field theory. In Theorem \ref{type-III} and Corollary \ref{type-III-corollary}, it is shown that any normal state of any local algebra violates KCBS inequality.
The following lemma plays an important role in the proof of Theorem \ref{type-III}.

\begin{lemma}
\label{kawahigashi}
\cite[Corollary 2.2]{kawahigashi2013normal}

Let $\mathfrak{N}$ be a type III factor on a separable Hilbert space, let $\mathfrak{M}$ be a finite dimensional C*-subalgebra of $\mathfrak{N}$, and let $\rho$ be a faithful normal state of $\mathfrak{N}$. For any finite set $\{ \psi_1, \dots , \psi_k \}$ of normal states of $\mathfrak{N}$, there exists an unitary operator $U \in \mathfrak{N}$ such that
\[ \psi_l(U^*AU)=\rho(A) \]
for any $A \in \mathfrak{M}$ and any $\psi_{l} \in \{ \psi_1, \dots , \psi_k \}$.
\end{lemma}

By using Lemma \ref{kawahigashi}, we show Theorem \ref{type-III}.

\begin{theorem}
\label{type-III}
Let $\mathfrak{N}$ be a type III factor on a separable Hilbert space $\mathcal{H}$. For any finite set $\{ \psi_1, \dots , \psi_k \}$ of normal states of $\mathfrak{N}$ and any real number $\epsilon \in (0, \sqrt{5}-2)$, there are projections $P_0, \dots P_4$ such that $P_i P_{i+1}=0$ for any $i \in \{0, \dots, 4 \}$, where $i+1$ is understood mod $5$, and
\[ \psi_l(P_{0}+P_{1}+P_{2}+P_{3}+P_{4}) \geq \sqrt{5} - \epsilon > 2, \]
for any $\psi_l \in \{ \psi_1, \dots , \psi_k \}$.
\end{theorem}

\begin{proof}
Let $\{ \psi_1, \dots , \psi_k \}$ be a finite set of of normal states of $\mathfrak{N}$ and let $\epsilon$ be a real number in $(0, \sqrt{5}-2)$.

Since $\mathfrak{N}$ is a type III factor, there is a system of matrix units $\{ V_{ij} \}_{i,j=1,2,3}$ such that $V_{ij}V_{kl}=\delta_{jk}V_{il}$, $V_{ij}^{*}=V_{ji}$, and $\sum_{i=1}^3 V_{ii} = I$, where $I$ is an identity operator on $\mathcal{H}$ \cite[Lemma 6.10]{maeda1989probability}.

We can construct projections $R_0, \dots , R_4$ in $\mathfrak{N}$ in a similar way to construct projections $R_{i}$ in Section \ref{finite-section}.
\[ R_0=\frac{1}{1+\cos (1/5)\pi} ( V_{11} + \sqrt{\cos (1/5)\pi} \cdot V_{13} + \sqrt{\cos (1/5)\pi} \cdot V_{31} + \cos (1/5)\pi \cdot V_{33} ), \]
\[ \begin{split}
R_1=\frac{1}{1+\cos (1/5)\pi} &( \cos^2 (4/5)\pi \cdot V_{11} + \sin (4/5)\pi \cdot \cos (4/5)\pi \cdot V_{12} \\
&+ \cos (4/5)\pi \cdot \sqrt{\cos (1/5)\pi} \cdot V_{13} +\sin (4/5)\pi \cdot \cos (4/5)\pi \cdot V_{21} \\
&+ \sin^2 (4/5)\pi \cdot V_{22} + \sin (4/5)\pi \cdot \sqrt{\cos(1/5)\pi} \cdot V_{23} \\
&+\cos (4/5)\pi \cdot \sqrt{\cos (1/5)\pi} \cdot V_{31} + \sin (4/5)\pi \cdot \sqrt{\cos (1/5)\pi} \cdot V_{32} \\
&+ \cos (1/5)\pi \cdot V_{33} ) ,
\end{split} \]
\[ \begin{split}
R_2=\frac{1}{1+\cos (1/5)\pi} &( \cos^2 (2/5)\pi \cdot V_{11} - \sin (2/5)\pi \cdot \cos (2/5)\pi \cdot V_{12} \\
&+ \cos (2/5)\pi \cdot \sqrt{\cos (1/5)\pi} \cdot V_{13} -\sin (2/5)\pi \cdot \cos (2/5)\pi \cdot V_{21} \\
&+ \sin^2 (2/5)\pi \cdot V_{22} - \sin (2/5)\pi \cdot \sqrt{\cos(1/5)\pi} \cdot V_{23} \\
&+\cos (2/5)\pi \cdot \sqrt{\cos (1/5)\pi} \cdot V_{31} - \sin (2/5)\pi \cdot \sqrt{\cos (1/5)\pi} \cdot V_{32} \\
&+ \cos (1/5)\pi \cdot V_{33} ),
\end{split} \]
\[ \begin{split}
R_3=\frac{1}{1+\cos (1/5)\pi} &( \cos^2 (2/5)\pi \cdot V_{11} + \sin (2/5)\pi \cdot \cos (2/5)\pi \cdot V_{12} \\
&+ \cos (2/5)\pi \cdot \sqrt{\cos (1/5)\pi} \cdot V_{13} +\sin (2/5)\pi \cdot \cos (2/5)\pi \cdot V_{21} \\
&+ \sin^2 (2/5)\pi \cdot V_{22} + \sin (2/5)\pi \cdot \sqrt{\cos(1/5)\pi} \cdot V_{23} \\
&+\cos (2/5)\pi \cdot \sqrt{\cos (1/5)\pi} \cdot V_{31} + \sin (2/5)\pi \cdot \sqrt{\cos (1/5)\pi} \cdot V_{32} \\
&+ \cos (1/5)\pi \cdot V_{33} ), 
\end{split} \]
\[ \begin{split}
R_4=\frac{1}{1+\cos (1/5)\pi} &( \cos^2 (4/5)\pi \cdot V_{11} - \sin (4/5)\pi \cdot \cos (4/5)\pi \cdot V_{12} \\
&+ \cos (4/5)\pi \cdot \sqrt{\cos (1/5)\pi} \cdot V_{13} -\sin (4/5)\pi \cdot \cos (4/5)\pi \cdot V_{21} \\
&+ \sin^2 (4/5)\pi \cdot V_{22} - \sin (4/5)\pi \cdot \sqrt{\cos(1/5)\pi} \cdot V_{23} \\
&+\cos (4/5)\pi \cdot \sqrt{\cos (1/5)\pi} \cdot V_{31} - \sin (4/5)\pi \cdot \sqrt{\cos (1/5)\pi} \cdot V_{32} \\
&+ \cos (1/5)\pi \cdot V_{33} ).
\end{split} \]
$R_0, \dots, R_{4}$ are projections in $\mathfrak{N}$, and
\begin{equation}
\label{KCBS-projection-type-III}
R_iR_{i+1}=0 \ \ \ (i \in \{0, \dots ,4 \}, \ \text{where} \ i+1 \  \text{is understood mod} \ 5).
\end{equation}
Let $\Psi$ be a unit vector in $V_{33}\mathcal{H}$. Then
\begin{equation}
\label{e2-type-III}
\langle \Psi, (R_0+R_1+R_2+R_3+R_4)\Psi \rangle = \frac{5 \cos (1/5)\pi}{1+\cos (1/5)\pi}=\sqrt{5}. 
\end{equation}

We will show that there is a faithful normal state $\rho$ such that
\begin{equation}
\label{e5-type-III}
\rho \left( \sum_{i=0}^4 R_i \right) \geq \sqrt{5} - \epsilon. 
\end{equation}
If a vector state of $\mathfrak{N}$ induced by $\Psi$ is faithful, it is proved because of Equation (\ref{e2-type-III}). Thus we assume that a vector state of $\mathfrak{N}$ induced by $\Psi$ is not faithful.

Since $\mathcal{H}$ is separable, there is a *-isomorphism $\pi$ such that a von Neumann algebra $\pi(\mathfrak{N})$ on a Hilbert space $\mathcal{K}$ admits a cyclic and separating vector \cite[Proposition 2.5.6]{bratteli1987operator}. By \cite[Theorem 2.5.31]{bratteli1987operator}, there is a unit vector $\Phi$ in $\mathcal{K}$ such that
\begin{equation}
\label{e3-type-III}
\langle \Psi, A \Psi \rangle = \langle \Phi, \pi(A) \Phi \rangle 
\end{equation}
for any $A \in \mathfrak{N}$.

Let $S$ be the support of the vector state of $\pi(\mathfrak{N})$ induced by $\Phi$. Then $\Phi$ is a separating vector for $S\pi(\mathfrak{N})S$. Let $\Phi'$ be a separating vector for $\pi(\mathfrak{N})$. Observe $S^{\perp}\Phi' \neq 0$ and define $\Phi^{\perp} := S^{\perp}\Phi'/\| S^{\perp}\Phi' \|$. Then $\Phi^{\perp}$ is a separating vector for $S^{\perp}\pi(\mathfrak{N})S^{\perp}$. Define a normal state $\omega$ of $\pi(\mathfrak{N})$ as
\begin{equation}
\label{e4-type-III}
\omega(B) := \left(1 - \frac{\epsilon}{\sqrt{5}} \right) \langle \Phi, B \Phi \rangle + \frac{\epsilon}{\sqrt{5}} \langle \Phi^{\perp}, B \Phi^{\perp} \rangle 
\end{equation}
for any operator $B \in \pi(\mathfrak{N})$.

Let $Q'$ be the support of $\omega$ and let $Q:= I - Q'$. Then 
\begin{equation}
0=\omega (Q)= \left(1 - \frac{\epsilon}{\sqrt{5}} \right) \langle \Phi, Q \Phi \rangle + \frac{\epsilon}{\sqrt{5}} \langle \Phi^{\perp}, Q \Phi^{\perp} \rangle .
\end{equation}
Since $ \langle \Phi, Q \Phi \rangle \geq 0$ and $\langle \Phi^{\perp}, Q \Phi^{\perp} \rangle \geq 0$, $ \langle \Phi, Q \Phi \rangle = \langle \Phi^{\perp}, Q \Phi^{\perp} \rangle = 0$.
$SQS \leq Q$ and $S^{\perp}QS^{\perp} \leq Q$ imply
\begin{equation}
\langle \Phi, SQS \Phi \rangle = \langle \Phi^{\perp},S^{\perp}QS^{\perp} \Phi^{\perp} \rangle = 0. 
\end{equation}
Thus $SQS=S^{\perp}QS^{\perp}=0$. 
\begin{equation}
\label{Q}
\begin{split}
Q &=Q^{2} \\
&=(SQS+SQS^{\perp}+S^{\perp}QS+S^{\perp}QS^{\perp})^{2} \\
&=(SQS^{\perp}+S^{\perp}QS)^{2} \\
&=SQS^{\perp}QS+S^{\perp}QSQS^{\perp}
\end{split}
\end{equation}
entails
\begin{equation}
\begin{split}
0 &=\omega(Q) \\
&=\omega(SQS^{\perp}QS+S^{\perp}QSQS^{\perp})  \\
&=\left(1 - \frac{\epsilon}{\sqrt{5}} \right) \langle \Phi, SQS^{\perp}QS \Phi \rangle + \frac{\epsilon}{\sqrt{5}} \langle \Phi^{\perp}, S^{\perp}QSQS^{\perp} \Phi^{\perp} \rangle 
\end{split}
\end{equation}

Since $ \langle \Phi, SQS^{\perp}QS \Phi \rangle \geq 0$ and $\langle \Phi^{\perp}, S^{\perp}QSQS^{\perp} \Phi^{\perp} \rangle \geq 0$, $SQS^{\perp}QS=S^{\perp}QSQS^{\perp}=0$. By Equation (\ref{Q}), $Q=0$. Therefore $\omega$ is a faithful normal state of $\pi(\mathfrak{N})$.

Let $\rho(A) := \omega (\pi(A))$ for any $A \in \mathfrak{N}$. Then $\rho$ is a faithful normal state of $\mathfrak{N}$, and
\begin{equation}
\label{e1-type-III}
\begin{split}
&\rho \left( \sum_{i=0}^{4} R_{i} \right) \\
&= \left(1 - \frac{\epsilon}{\sqrt{5}} \right) \Bigg\langle \Psi, \left( \sum_{i=0}^{4} R_{i} \right) \Psi \Bigg\rangle + \frac{\epsilon}{\sqrt{5}} \Bigg\langle \Phi^{\perp},  \pi \left( \sum_{i=0}^{4} R_{i} \right) \Phi^{\perp} \Bigg\rangle  \\
&\geq \sqrt{5} - \epsilon
\end{split}
\end{equation}
by Equations (\ref{e2-type-III}), (\ref{e3-type-III}), and (\ref{e4-type-III}). Thus Inequality (\ref{e5-type-III}) was shown.

Let $\mathfrak{M}$ be a C*-subalgebra of $\mathfrak{N}$ generated by $\{ V_{ij} \}_{i,j=1,2,3}$. Then $\sum_{i=1}^{5} R_{i}  \in \mathfrak{M}$. By Lemma \ref{kawahigashi} and Inequality (\ref{e1-type-III}), there is an unitary operator $U \in \mathfrak{N}$ such that
\begin{equation}
\label{e0-type-III}
\psi_{l} \left( U^{*} \left( \sum_{i=0}^{4} R_{i} \right) U \right) =\rho \left( \sum_{i=0}^{4} R_{i} \right) \geq \sqrt{5} - \epsilon
 \end{equation}
for any $\psi_{l} \in \{ \psi_1, \dots , \psi_k \}$.

Let $P_{i} := U^{*}R_{i}U$ for any $i \in \{0, \dots ,4 \}$. Then $P_{i}P_{i+1}=0$ for any $i \in \{0, \dots ,4 \}$ by Equation (\ref{KCBS-projection-type-III}). Inequality (\ref{e0-type-III}) and $\epsilon \in (0, \sqrt{5}-2)$ entail
\begin{equation}
\psi_{l} \left( \sum_{i=0}^{4} P_{i}  \right) \geq \sqrt{5} - \epsilon > 2
\end{equation}
for any $\psi_{l} \in \{ \psi_1, \dots , \psi_k \}$.
\end{proof}

According to Theorem \ref{type-III}, if we fix a finite set $\{ \psi_1, \dots, \psi_k \}$ of normal states, there are projections $P_0, \dots, P_4$ with which any normal state $\psi_l \in \{ \psi_1, \dots, \psi_k \}$ violates KCBS inequality. $k$ can be any number, say, $100$ trillion. Thus, roughly speaking, KCBS inequality is `almost' state-independent in algebraic quantum field theory.

As an immediate consequence of Theorem \ref{type-III}, we get

\begin{corollary}
\label{type-III-corollary}
Let $\mathfrak{N}$ be a type III factor on a separable Hilbert space. Any normal state of $\mathfrak{N}$ violates KCBS inequality.
\end{corollary}

Since KCBS inequality is satisfied by any noncontextual hidden variable theory, Corollary \ref{type-III-corollary} shows a contradiction between predictions of algebraic quantum field theory and noncontextual hidden variable theories.


\section{Conclusion}
The noncontextuality condition states that a value of any observable is independent of which other compatible observable is measured jointly with it. If we assume that there is a noncontextual hidden variable theory, KCBS inequality holds. Thus, its violation shows a contradiction between predictions of quantum theory and noncontextual hidden variable theories. 

In Sections \ref{finite-section} and \ref{type-III-section}, we examined KCBS inequality in quantum mechanics of finite degrees of freedom and algebraic quantum field theory. It was shown that a tracial state does not violate KCBS inequality in the case of quantum mechanics of finite degrees of freedom in Theorem \ref{trace-state} and the following part. But the theorem cannot apply to a set of all bounded operators on an infinite dimensional Hilbert space because there is no tracial state in this case. It is important whether all normal states violate KCBS inequality or not in this case since it describes quantum mechanics of finite degrees of freedom. It is an open problem.

In Theorem \ref{type-III} and Corollary \ref{type-III-corollary}, we examine KCBS inequality in algebraic quantum field theory. Because a tracial states and a pure normal state do not exist in algebraic quantum field theory, the argument of KCBS inequality in quantum mechanics of finite degrees of freedom cannot apply to algebraic quantum field theory. In Corollary \ref{type-III-corollary}, it is shown that any normal state violates it in the case of algebraic quantum field theory. It is a difference between quantum mechanics of finite degrees of freedom and algebraic quantum field theory from a point of view of KCBS inequality.

\section*{Acknowledgments}
The authors thank two referees for helpful comments for an earlier version of this paper. The author is supported by the JSPS KAKENHI No.15K01123 and No.23701009.

\bibliographystyle{jplain}
\bibliography{kitajima}

\end{document}